\documentclass[runningheads]{llncs}

\usepackage{graphicx}
\usepackage{pifont}
\usepackage{verbatim}
\usepackage{amsmath,amsfonts}
\usepackage{enumerate,url}
\usepackage{paralist}
\usepackage{subcaption}
\usepackage{todonotes}
\usepackage{wrapfig}

\begin{document}

\newif\ifarxiv 
\arxivtrue
\newif\iffigdir 
\figdirfalse

\let\doendproof\endproof
\renewcommand\endproof{~\hfill\qed\doendproof}

\makeatletter
    \renewcommand*{\@fnsymbol}[1]{\ensuremath{\ifcase#1\or *\or \dagger\or \S\or
       \mathsection\or \mathparagraph\or \|\or **\or \dagger\dagger
       \or \ddagger\ddagger \else\@ctrerr\fi}}
\makeatother

\newcommand*\samethanks[1][\value{footnote}]{\footnotemark[#1]}
\newcommand{\elone}{grid}

\title{Grid-Obstacle Representations with Connections to Staircase Guarding}

\titlerunning{Grid-Obstacle Representations and Staircase Guarding}

\author{
Therese Biedl$^1$\and
Saeed Mehrabi$^1$}
\authorrunning{T. Biedl, and S. Mehrabi}
\institute{
$^1$~Cheriton School of Computer Science\\
University of Waterloo, Waterloo, Canada.\\
\email{\{biedl, smehrabi\}@uwaterloo.ca}
}

\newcommand{\keywords}[1]{\par\addvspace\baselineskip
\noindent\keywordname\enspace\ignorespaces#1}

\maketitle

\begin{abstract}
In this paper, we study \elone-obstacle representations of graphs where
we assign grid-points to vertices and define
obstacles such that an edge exists if and only if an
$xy$-monotone grid path connects the two endpoints without hitting
an obstacle or another vertex.  It was previously argued that
all planar graphs have a \elone-obstacle representation in 2D,
and all graphs have a \elone-obstacle representation in 3D.
In this paper, we show that such constructions are possible with
significantly smaller grid-size than previously achieved. 
Then we study the variant where vertices are not blocking, and
show that then \elone-obstacle representations exist for bipartite graphs.
The latter has applications in so-called \emph{staircase guarding} of
orthogonal polygons; using our \elone-obstacle representations, we show that staircase guarding is \textsc{NP}-hard in 2D. 
\end{abstract}

\section{Introduction}
\label{sec:introduction}
Recently, Bishnu et al.~\cite{BGMMP} initiated the study of
{\em \elone-obstacle representations}.  Here the vertices of a
graph $G=(V,E)$ are mapped to points in an integer grid, and other grid-points are marked
as {\em  obstacles} in such a way that $(v,w)$ is an
edge of $G$ if and only if there exists an $xy$-monotone path 
in the grid from $v$ to $w$ that contains no obstacle-point and no
point that belongs to some vertex $\neq v,w$.  
See also Fig.~\ref{fig:obst_repr}. This is a special case 
of a more general problem, which
asks for placing points and obstacles in the plane such that an 
edge $(v,w)$ exists if and only if there is a shortest path (in some distance
metric) from $v$ to $w$ that does not intersect obstacles.  See
also Alpert et al.~\cite{AKL10}, who initiated the study of obstacle
numbers, and~\cite{DM15} and the references therein for more recent developments.

Bishnu et al.~\cite{BGMMP} showed that any planar graph has a \elone-obstacle representation in 2D, and every graph has a \elone-obstacle representation in 3D. 
The main idea was
to use a straight-line drawing, and then approximate it by putting a sufficiently
fine grid around it that consists of obstacles everywhere except near the edge.
The analysis of how fine a grid is required is not straightforward; Bishnu et al.~claimed
that in 2D an $O(n^2)\times O(n^2)$-grid is sufficient. They did not give bounds for
the size needed in 3D (but it clearly is polynomial and at least $\Omega(n^2)$ in each dimension).
Pach showed that not all bipartite graphs have \elone-obstacle representations in 2D~\cite{P16}.

In this paper, we improve the grid-size bounds of \cite{BGMMP}.  In particular, rather than
converting a straight-line drawing directly into a \elone-obstacle representation,
we first convert it into a visibility representation or an orthogonal drawing that
has special properties, but resides in a linear-size grid. This can then
be easily converted to a \elone-obstacle representation.  Thus we obtain
2D \elone-obstacle representations for planar graphs in an $O(n)\times O(n)$-grid,
and 3D \elone-obstacle representations for all graphs in an $O(n)\times O(n)\times O(n)$-grid.

We then discuss the case with the restriction that vertices act as obstacles for edges not incident to them, and show that sometimes this restriction can be dropped. We hence obtain {\em non-blocking \elone-representations} in 2D for all planar bipartite graphs and in 3D for arbitrary bipartite graphs.

The latter has applications: we can use the constructions for hardness proofs for a polygon-guarding problem. A point guard $g$ is said to {\em staircase guard} (or \emph{$s$-guard} for short) a point $p$ inside an orthogonal polygon $P$ if $p$ can be reached from $g$ by a {\em staircase}; that is, an orthogonal path inside $P$ that is both $x$- and $y$-monotone. In the \emph{$s$-guarding problem}, the objective is to guard an orthogonal polygon with the minimum number of $s$-guards. Motwani et al.~\cite{MotwaniRS90} proved that $s$-guarding is polynomial on simple orthogonal polygons. Gewali and Ntafos~\cite{GewaliN92} proved that the problem is \textsc{NP}-hard in 3D; 
since they reduce from vertex cover in graphs with maximum degree 3 this in fact implies \textsc{APX}-hardness in 3D~\cite{AlimontiK00}. To our knowledge, however, the complexity was open for 2D polygons with holes. Using non-blocking \elone-representations, we show that it is \textsc{NP}-hard.

\section{2D Grid-Obstacle Representations}
\label{sec:2D}
Let $G=(V,E)$ be a planar graph.  
To build a \elone-obstacle representation, we use a {\em visibility
representation}  where every vertex is represented by a {\em bar} (a horizontal line segment), 
and every edge is represented by a vertical line segment between the bars corresponding to the endpoints of the edge~\cite{RT86},\cite{TT86},\cite{Wis85}. We need here a construction with a special property, which can easily be achieved by ``shifting'' where the edges attach at the vertices (see 
\ifarxiv
the appendix
\else
the full paper
\fi
for a direct proof). See also Fig.~\ref{fig:specialVR}.

\begin{lemma}
\label{lem:specialVR}
Every planar graph has a visibility representation in an $O(n)\times O(n)$-grid for which any vertex-bar can be split into a left and
right part such that all downward edges attach on the left and all upward edges attach on the right.
\end{lemma}

\begin{figure}[t]
\hspace*{\fill}
\hspace*{\fill}
\begin{subfigure}[b]{0.4\linewidth}
\iffigdir
	\includegraphics[width=0.8\linewidth,page=3]{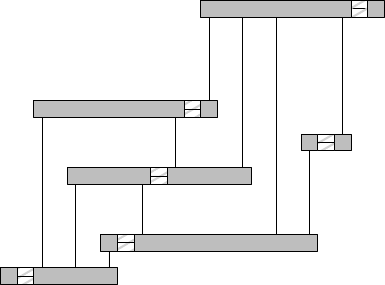}
\else
	\includegraphics[width=0.8\linewidth,page=3]{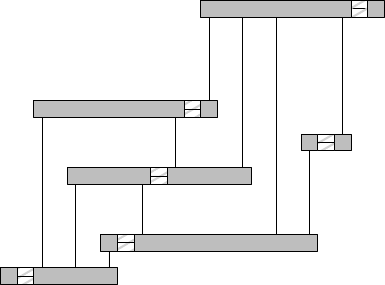}
\fi
	\caption{}
	\label{fig:specialVR}
\end{subfigure}
\hspace*{\fill}
\begin{subfigure}[b]{0.4\linewidth}
\iffigdir
	\includegraphics[width=0.8\linewidth,page=5]{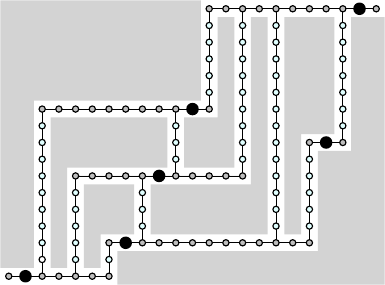}
\else
	\includegraphics[width=0.8\linewidth,page=5]{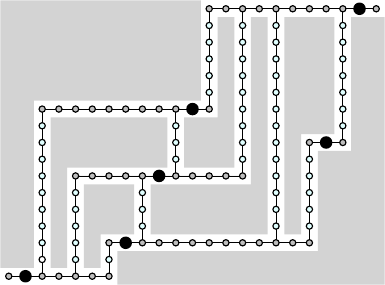}
\fi
	\caption{}
	\label{fig:obst_repr}
\end{subfigure}
\hspace*{\fill}
\caption{A special visibility representation
gives a \elone-obstacle representation.  }
\label{fig:planar2D}
\end{figure}

Now convert such a visibility representation into a \elone-obstacle representation.  First, 
double the grid so that no two grid-points on edge-segments or vertex-bars 
are adjacent unless the corresponding graph-elements were.
For each vertex $v$, assign as vertex-point some grid-point that lies between the two 
parts of the bar of $v$; this exists since we doubled the grid.  The obstacles 
consist of all grid points that are not on some edge segment
or vertex bar. Clearly, the representation is in an $O(n)\times O(n)$-grid. 
In 
\ifarxiv
the appendix,
\else
the full paper,
\fi
we show that this is a \elone-obstacle representation, and so we have:
\begin{theorem}
\label{thm:2Dobstacle}
\label{thm:planar}
Every planar graph has a 2D \elone-obstacle representation in an $O(n)\times O(n)$-grid.
\end{theorem}

One can easily argue that any straight-line drawing of a planar graph of height $H$
can be converted into a visibility representation of height $2H$ and width $O(n)$ 
(see also~\cite{Biedl14}).  Then we can apply the same approach as above.
Based on drawings for trees~\cite{CrescenziBP92},
outer-planar graphs \cite{Biedl11} and series-parallel graphs~\cite{Biedl11}, we hence get:

\begin{corollary}
Every tree and every outer-planar graph has a 2D \elone-obstacle representation in an 
$O(\log n)\times O(n)$-grid. Every series-parallel graph has a 2D \elone-obstacle
representation in an $O(\sqrt{n})\times O(n)$-grid.
\end{corollary}

\section{3D Grid-Obstacle Representation}
\label{sec:3D}
In this section, we argue that a similar (and even simpler) construction
gives a \elone-obstacle representation in 3D.  We obtain this by building
an orthogonal representation first that has special properties.  This
representation is not quite a graph drawing, because edges may overlap;
this will not create problems for the obstacle representation later.

Enumerate the vertices as $v_1,\dots,v_n$ in arbitrary order.  Place $v_i$
at $(i,i,i)$.  To draw an edge $(v_i,v_j)$ with $i<j$, we use 
the path 
$(i,i,i)-(j,i,i)-(j,i,j)-(j,j,j)$
along the cube spanned between the two points.  See Fig.~\ref{fig:all3D}.
Observe that all edges $(v_h,v_i)$ with $h<i$ reach $v_i$ from the $y^-$-side and that all edges $(v_i,v_j)$ with $i<j$ leave $v_i$ at the $x^+$-side.
Edges incident to $v_i$ may overlap along these two sides, but otherwise
there are no overlaps or crossings in the drawing.  Also, we clearly reside in an
$n\times n\times n$-grid.

\begin{figure}[t]
\hspace*{\fill}
\iffigdir
	\includegraphics[width=0.25\linewidth,page=1]{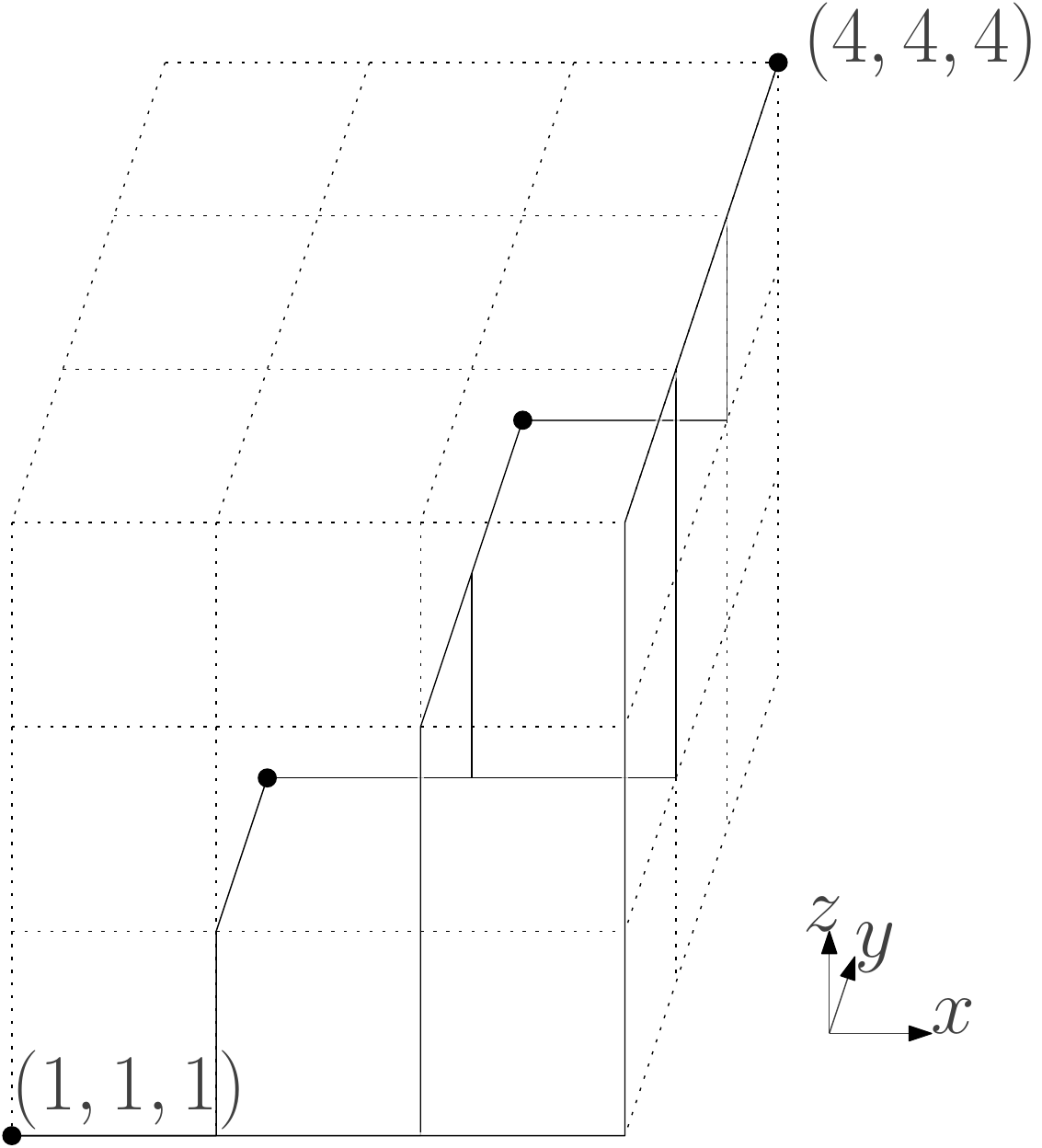}
\else
	\includegraphics[width=0.25\linewidth,page=1]{all3D.pdf}
\fi
\hspace*{\fill}
\iffigdir
	\includegraphics[width=0.25\linewidth,page=2]{figures/all3D.pdf}
\else
	\includegraphics[width=0.25\linewidth,page=2]{all3D.pdf}
\fi
\hspace*{\fill}
\caption{A 3D orthogonal representation of $K_4$, and converting it into 
a \elone-obstacle representation.  All grid-points that are not shown
are blocked by obstacles.}
\label{fig:all3D}
\end{figure}

Now double the grid, then cover any grid-point by an obstacle unless it is used by a vertex or an edge.   One can easily argue that the result is
a \elone-obstacle representation (see 
\ifarxiv
the appendix
\else
the full paper
\fi
for a formal proof), and we have:
\begin{theorem}
\label{thm:3Dobstacle}
Every graph has a 3D \elone-obstacle representation in an $O(n)\times O(n)\times O(n)$-grid.
\end{theorem}

Notice that the obstacle in this case can be made
to be just one polyhedron (albeit of high genus).

\section{Non-Blocking Grid-Obstacle Representations}
\label{sec:notVertices}
In our definition of \elone-obstacle representation, we required that the
grid point of any vertex  $v$ acts as an obstacle to any other path. 
The main reason for this is that otherwise paths could ``seep
through'' a vertex, creating unwanted adjacencies.  In this section, we 
consider {\em non-blocking \elone-obstacle representations}, which means
that vertices do not act as obstacles.

\subsection{Planar bipartite graphs}
\label{subsec:planarBGraphs}
We first give an algorithm for non-blocking
\elone-obstacle representation of planar
bipartite graphs. It is known that any such graph $G=(A\cup B,E)$
has an {\em HH-drawing} \cite{BiedlKM98}, i.e., a planar drawing where
all vertices in $A$ have positive $y$-coordinate,
all vertices in $B$ have negative $y$-coordinate,
every edge is drawn with at most one bend, 
and all bends have $y$-coordinate 0.
See also Fig.~\ref{fig:bipartite}.

In particular, we know that every edge is drawn $y$-monotonically.
Any such drawing can be converted into a visibility
representation~\cite{Biedl14} where the $y$-co\-or\-di\-nate of every
vertex is unchanged. So we obtain:

\begin{lemma}
\label{lem:hhDrawing}
Let $G=(A\cup B,E)$ be a planar bipartite graph.
Then, there exists a visibility representation of $G$ such that
all vertices in $A$ have only neighbours below,
and all vertices in $B$ have only neighbours above.
\end{lemma}

Now create an obstacle representation as before by 
doubling the grid, and placing obstacles at all grid-points that are
not used by the drawing.  Place each vertex $a\in A$ at
the rightmost grid-point of the bar of $a$, and each $b\in B$
at the leftmost grid-point of the bar of $b$.
One easily verifies that this is a non-blocking \elone-obstacle 
representation:  For each vertex $a$ in $A$, no $xy$-monotone path can go 
through the grid-point of $a$ without ending there, because no grid-point 
higher than $a$ can be reached when going through $a$.  Similarly one
argues for $B$, and so  we have:

\begin{theorem}
\label{thm:planarBipartite}
Every planar bipartite graph has a non-blocking \elone-obstacle
representation in an $O(n)\times O(n)$-grid.
\end{theorem}

\begin{figure}[ht]
\hspace*{\fill}
\iffigdir
	\includegraphics[width=0.3\linewidth,page=1]{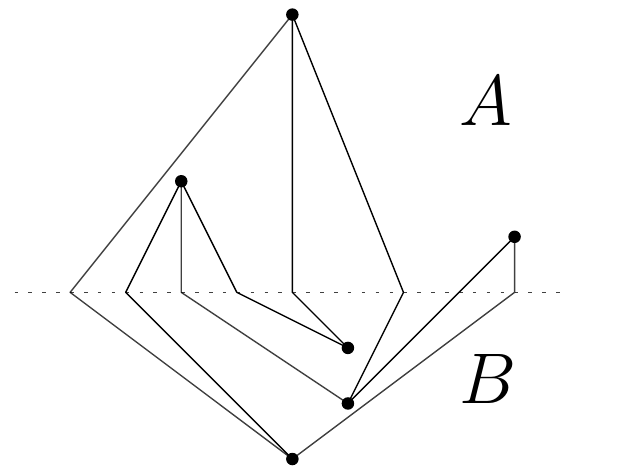}
	\hspace*{\fill}
	\includegraphics[width=0.3\linewidth,page=2]{figures/bipartite2D.pdf}
	\hspace*{\fill}
	\includegraphics[width=0.3\linewidth,page=3]{figures/bipartite2D.pdf}
\else
	\includegraphics[width=0.3\linewidth,page=1]{bipartite2D.pdf}
	\hspace*{\fill}
	\includegraphics[width=0.3\linewidth,page=2]{bipartite2D.pdf}
	\hspace*{\fill}
	\includegraphics[width=0.3\linewidth,page=3]{bipartite2D.pdf}
\fi
\hspace*{\fill}
\caption{An HH-drawing of a planar bipartite graph, and
converting it to a non-blocking \elone-obstacle representation.}
\label{fig:bipartite}
\end{figure}

\subsection{Application to staircase guarding}
Recall that the $s$-guarding problem consists of finding the minimum
set $S$ of points in a given orthogonal polygon $P$ such that for any $q\in P$
there exists a $p\in S$ that is connected to $q$ via a staircase inside $P$.
Using non-blocking \elone-obstacle representations, we can show:

\begin{theorem}
\label{thm:sGuarding}
$s$-guarding is \textsc{NP}-hard on orthogonal polygons with holes.
\end{theorem}
\begin{proof}
We reduce from minimum dominating set, i.e., the problem of finding
a set $D$ of vertices in a graph such that every vertex is either in 
$D$ or has a neighbor in $D$.  This is \textsc{NP}-hard, even
 on planar bipartite graphs \cite{CCJ90}. 
Given a planar bipartite graph $G=(A\cup B, E)$, construct the
non-blocking \elone-obstacle representation $\Gamma$ 
from Theorem~\ref{thm:planarBipartite}.
Let $P'$ consist of all
unit squares ({\em pixels}) 
around grid-points that are {\em not} in an obstacle.  
The obstacles of $\Gamma$ become holes in $P'$. 
Now for any vertex $a\in A$ 
extend the bar of $a$ slightly
rightward beyond the last edge, and for every $b\in B$ extend the bar leftward
beyond the last edge.
Finally, at every edge $e$, attach two ``spirals'' on the left and right side 
of its vertical segment; the one on
the left curls upward while the on the right curls downward.
See Fig.~\ref{fig:sguarding}.
These spirals are small enough that they fit within the
holes of $P'$, without overlapping other parts of $P'$ or each other.
We show in 
\ifarxiv
the appendix
\else
the full paper
\fi
that $G$ has a dominating set of size $k$
if and only if this polygon can be $s$-guarded with $2|E|+k$ guards.  
This proves the theorem.
\begin{figure}[h]
\hspace*{\fill}
\begin{subfigure}{0.30\linewidth}
\iffigdir
	\includegraphics[width=0.99\linewidth,page=1]{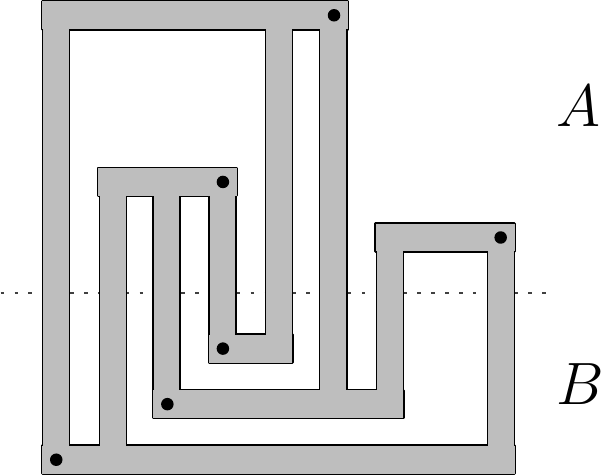}
\else
	\includegraphics[width=0.99\linewidth,page=1]{sGuards.pdf}
\fi
\caption{}
\end{subfigure}
\hspace*{\fill}
\begin{subfigure}{0.34\linewidth}
\iffigdir
	\includegraphics[width=0.99\linewidth,trim=0 3 0 12,clip,page=2]{figures/sGuards.pdf}
\else
	\includegraphics[width=0.99\linewidth,trim=0 3 0 12,clip,page=2]{sGuards.pdf}
\fi
\caption{}
\end{subfigure}
\hspace*{\fill}
\caption{The polygon for the graph in Fig.~\ref{fig:bipartite}, and gadgets that we attach.}
\label{fig:sguarding}
\label{fig:sGuarding}
\end{figure}
\end{proof}

\subsection{3D grid-obstacle representation of bipartite graphs}
In 3D, all bipartite graphs have a non-blocking \elone-obstacle representation:
Enumerate the vertices as $A=\{a_1,\dots,a_\ell\}$ and $B=\{b_1,\dots,b_k\}$.
Place a point for vertex $a_i$ at
$(0,i,0)$ and a point for vertex $b_j$ at $(j,0,1)$.  Route each
edge $(a_i,b_j)$ as the orthogonal path
$
(0,i,0)-(j,i,0)-(j,i,1)-(j,0,1),
$
and observe that two paths overlap in the $x^+$-direction at $a_i$ 
if they both begin at $a_i$, or overlap in the $y^+$-direction at $b_j$
if they both end at $b_j$, but otherwise there is no overlap.
Now obtain the \elone-obstacle representation as before by 
doubling the grid and making grid-points obstacles unless they are
used by vertices and edge-paths.
As before one argues that this is indeed a non-blocking
\elone-obstacle representation and so we have:

\begin{theorem}
Every bipartite graph has a 3D non-blocking \elone-obstacle representation.
\end{theorem}
\begin{figure}[t]
\hspace*{\fill}
\iffigdir
	\includegraphics[width=0.4\linewidth,page=3,trim=0 100 0 0,clip]{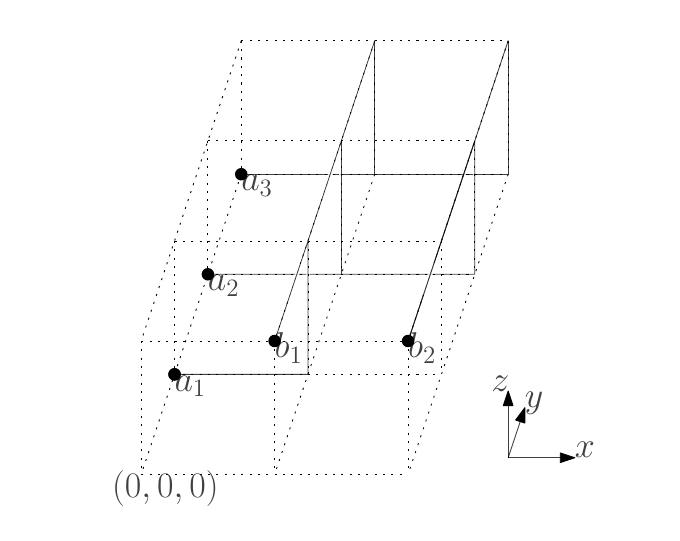}
	\hspace*{\fill}
	\includegraphics[width=0.4\linewidth,page=4,trim=0 100 0 0,clip]{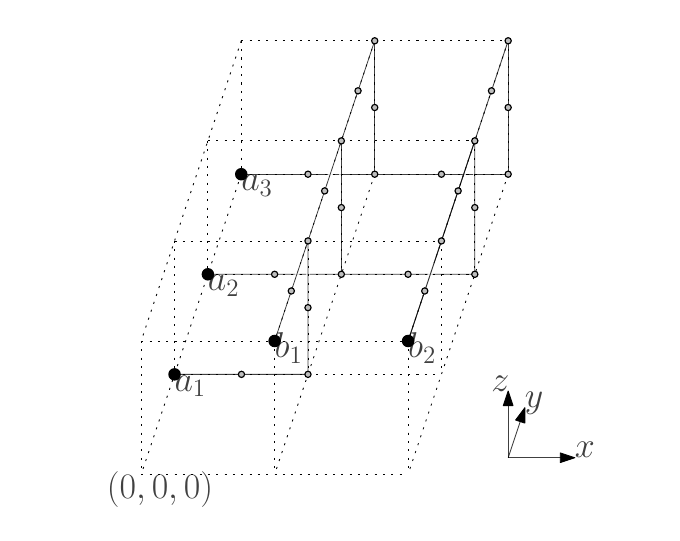}
\else
	\includegraphics[width=0.4\linewidth,page=3,trim=0 100 0 0,clip]{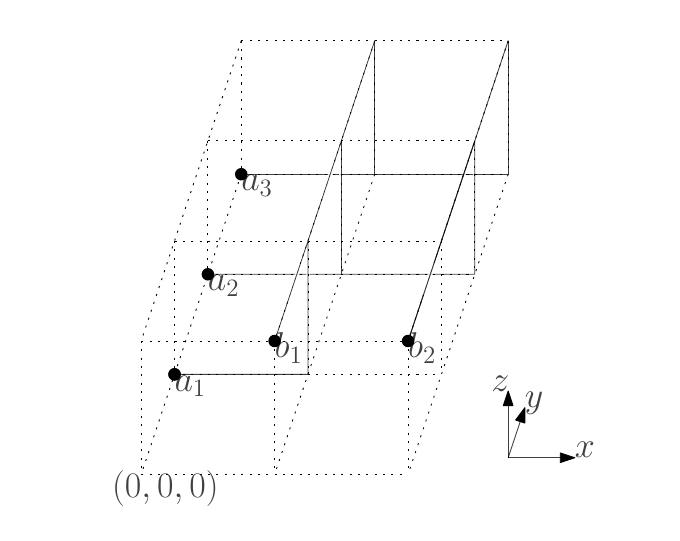}
	\hspace*{\fill}
	\includegraphics[width=0.4\linewidth,page=4,trim=0 100 0 0,clip]{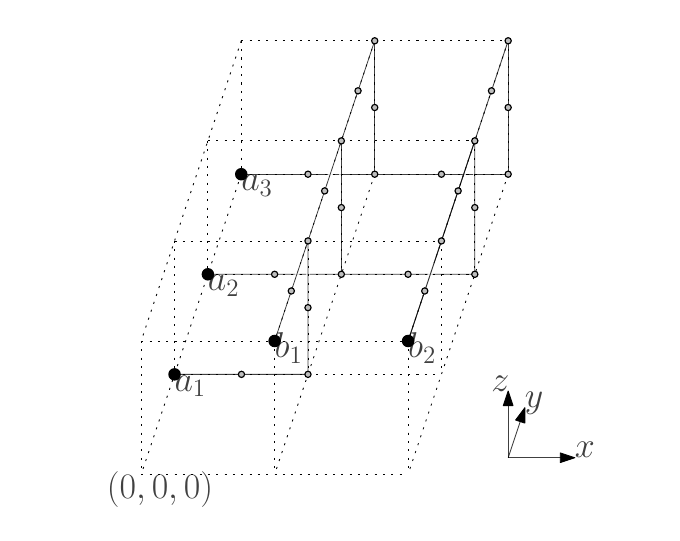}
\fi
\hspace*{\fill}
\caption{A 3D orthogonal representation of $K_{2,3}$, and converting it into
a \elone-obstacle representation.  Grid-points not shown are covered by obstacles.}
\label{fig:all3Dbipartite}
\end{figure}

\section{Conclusion}
\label{sec:conclusion}
In this paper, we studied \elone-obstacle representations. We gave constructions with smaller grid-size for planar graphs in 2D and all graphs in 3D.  If the graph is bipartite then we can construct representations where vertices are not considered obstacles.  We used these types of representation to prove NP-hardness of the $s$-guarding problem in 2D polygons with holes.

It remains open
whether an asymptotically smaller grid and/or fewer obstacles might be enough.
If we allow obstacles to be polygons rather than grid-points,
we use (in Theorems~\ref{thm:planar} and \ref{thm:planarBipartite})
one obstacle per face of the planar graph, or $\Theta(n)$ in total.
For \elone-obstacle representations that use straight-line segments, rather than
$xy$-monotone grid-paths, significantly fewer obstacles suffice~\cite{DM15}.
Can we create \elone-obstacle representations with $o(n)$ obstacles, at least
for some subclasses of planar graphs? Another direction for future work would be to find other classes of graphs for which we can construct non-blocking \elone-obstacle representations. Does this exist for all planar graphs in 2D?

\newpage
\bibliographystyle{splncs03}
\bibliography{ref}

\ifarxiv
\newpage
\begin{appendix}
\section{ Missing Proofs}

\paragraph{Proof of Lemma~\ref{lem:specialVR}}

Take a planar straight-line drawing of $G$ that has height $O(n)$ and where no edge is drawn 
horizontally. (For example, the drawing of de Fraysseix et al.~\cite{FPP90} is easily seen 
to achieve this if we modify the placement of the initial triangle $v_1,v_2,v_3$.) 
Direct the edges from the lower endpoint to the higher endpoint.  
Now split each vertex $v$ into two adjacent vertices $v^\ell$ and $v^r$, 
where $v^\ell$ is adjacent to all incoming edges of $v$ and $v^r$ is adjacent to all
outgoing edges.  Double the height of the drawing by inserting a new row after
each existing one.  Place $v^r$ to the right of $v$, and re-route all outgoing edges of $v$ to leave from $v^r$ instead, by adding a bend in the row above $v$. See Fig.~\ref{fig:splitDrawing}.

\begin{figure}[t]
\hspace*{\fill}
\begin{subfigure}[b]{0.4\linewidth}
\iffigdir
	\includegraphics[width=0.94\linewidth,page=1]{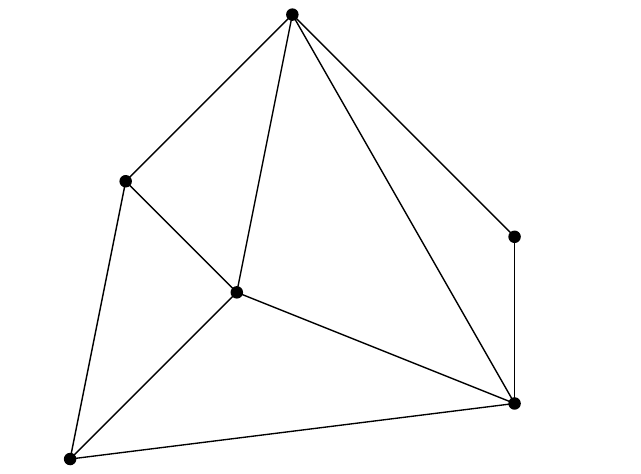}
\else
	\includegraphics[width=0.94\linewidth,page=1]{planar2D.pdf}
\fi
	\caption{}
\end{subfigure}
\hspace*{\fill}
\begin{subfigure}[b]{0.4\linewidth}
\iffigdir
	\includegraphics[width=0.94\linewidth,page=2]{figures/planar2D.pdf}
\else
	\includegraphics[width=0.94\linewidth,page=2]{planar2D.pdf}
\fi
	\caption{}
\end{subfigure}
\hspace*{\fill}
\caption{Split the vertices of a planar drawing. Converting this into
a visibility representation gives the one shown in Fig.~\ref{fig:planar2D}.}
\label{fig:splitDrawing}
\end{figure}

Observe that all edges are drawn $y$-monotonically, and $(v^\ell,v^r)$ is drawn horizontally
for all vertices $v$. Now we convert this drawing, using the method of~\cite{Biedl14}, into a visibility 
representation such that the $y$-coordinates and the order of vertices within a row is unchanged.  
This has again height $O(n)$, and (after deleting empty columns) it also has width $O(m+n)=O(n)$~\cite{Biedl14}.  The two vertices $v^\ell$ and $v^r$ that replaced
vertex $v$ were consecutive in one row, hence in the visibility representation they are also
consecutive (and the edge between them is horizontal) and we can simply recombine them to 
obtain one horizontal segment for vertex $v$.  
This gives the required visibility representation.

\paragraph{Proof of Theorem~\ref{thm:2Dobstacle}}
We must argue that the constructed representation is a \elone-obstacle
representation. For any vertex $v$, let $v^\ell$ and $v^r$ be the left
and the right part of the bar of $v$ where downward/upward edges attach.
Consider an edge $(v_i,v_k)$ where $v_i$ has the smaller $y$-coordinate.
In the visibility representation, the corresponding
segment attaches at $v_i^r$ and $v_k^\ell$.  In the obstacle
representation, we can hence go rightward from the grid-point of $v_i$ along edge $(v_i^\ell,v_i^r)$
and the segment of $v_i^r$, 
then up along the segment of $(v_i^r,v_k^\ell)$, then rightward along $v_k^\ell$ and the
edge $(v_k^\ell,v_k^r)$ to reach $v_k$ along an $xy$-monotone grid-path.

On the other hand, assume that $v_i$ and $v_k$ have an $xy$-monotone grid-path $\pi$ between them in the
obstacle representation.  No two vertex-points have the same $y$-coordinate,  so we may assume
that the point of $v_i$ is lower.  If $\pi$ left $v_i$   on the left side, then it would have
to continue downward from there, which contradicts monotonicity.  So $\pi$ leaves $v_i$ on the
right side.  From there, it can only go upward along some edge $(v_i,v_j)$ and reach
the segment of $v_j^\ell$.  All edges attaching here go downward, which $\pi$ cannot use by
monotonicity.  So $\pi$ must continue to the grid-point of $v_j$.  Here $\pi$ is obstructed if 
$v_j\neq v_k$, so we must have $v_j=v_k$ and $(v_i,v_k)$ is an edge as desired.

\paragraph{Proof of Theorem~\ref{thm:3Dobstacle}}
We must argue that the constructed representation is a \elone-obstacle representation. Clearly, for any edge $(v_i,v_k)$ we can find an $xyz$-monotone path by walking along the route of $(v_i,v_k)$. Vice versa, if there is a monotone path $\pi$ from $v_i$ to $v_k$ with (say) $i<k$,
then it must connect $(i,i,i)$ to $(k,k,k)$ and so be going in positive direction.
Thus it must leave $v_i$ on the $x^+$-side.  From here the only option is to
continue in $z^+$-direction starting at some point $(i,j,i)$.  This necessarily
leads to $(i,j,j)$, since there are no other adjacent unobstructed grid-points.  From there
the only positive direction possible is to go to $(j,j,j)$.  But then 
$v_j$ blocks the path, so we must have $v_j=v_k$ and $(v_i,v_k)$ is an edge 
as desired.

\paragraph{Proof of Theorem~\ref{thm:sGuarding}}
We aim to show that $G$ has a dominating set of size $k$ if and only if $P$ can be guarded by $k+2|E|$ $s$-guards.  Recall that in $P$ every vertex corresponds to a bar (of the visibility representation of $G$) and every edge $e$ corresponds to a channel (along the vertical segment
that represented $e$).  Also, for each vertex $u$ we attached
an ``end-pixel'' $\psi_u$ that is beyond all attachment points of 
all edge-channels; this is on the right end of the bar if $u\in A$ and on the left end if 
$u\in B$. These are marked by black dots in Fig.~\ref{fig:sGuarding}(b)).

For any edge $(a,b)$, we also attached two spiral-gadgets to the edge-channel of $(a,b)$. 
In any such spiral $\sigma$, there are two crucial places.
One is the ``tail-pixel'' $\psi_\sigma$ at the end of the spiral
(marked by a circle in Fig.~\ref{fig:sGuarding}(b)), and the other is the 
line segment $s_\sigma$ that marks the boundary of points that can $s$-guard
$\psi_\sigma$ (marked by a cross in Fig.~\ref{fig:sGuarding}(b)).

Consider a dominating set $D$ of size $k$ in $G$ and define a set $S$ of
points in $P$ as follows.  For each $u\in D$,  add an arbitrary point
$p(u)$ of the end-pixel $\psi_u$ to $S$.  Observe that $p(u)$ guards $\psi_u$
as well as $\psi_v$ of any vertex $v$ for
which $(u,v)$ is an edge.
Secondly, for each edge-spiral $\sigma$, add the point $x_\sigma$
marked with a cross in Fig.~\ref{fig:sGuarding}.  Note that the point $x_\sigma$ of
the left spiral $\sigma$ at an edge $(a,b)$ 
can $s$-guard all of $\sigma$, the bottom half of the edge-channel for $(a,b)$
and everything of the vertex-bar of $b$ (the lower endpoint)
that is to the right of where the
edge-channel attaches.  Similarly, the point in the right spiral can see
the top half of the edge-channel and everything of the vertex-bar of $a$
to the left of where the edge-channel attaches. All the spiral-guards 
together hence cover everything except the end-pixels, but those are guarded
by the points added due to dominating set $D$.  So the chosen $k+2|E|$ points
guard everything.

Conversely, let $S$ be a set of $k+2|E|$ points in $P$ that $s$-guard all of $P$.
For each spiral-gadget $\sigma$, there must exist some guard $s\in S$ that 
guards the tail-pixel.  This guard must lie on segment $s_\sigma$ (or even
closer to the tail-pixel), and as one easily verifies, cannot see any point
in an end-pixel, or any tail-pixel of any other spiral-gadget.  Therefore,
there are $2|E|$ such points in total (call them $M$), none of which guards 
an end-pixel.

This leaves at most $k$ guards that $s$-guard all end-pixels.
Define $D$ as follows.  For any vertex $v$, if some point in $S$ lies
in the vertex-bar of $v$, add $v$ to $D$.   For any edge $e=(u,v)$,
if some point in $S$ $s$-guards some end-pixel and lies on the edge-channel 
or a spiral-gadget of $e$, then arbitrarily
add one of $u$ and $v$ to $D$.  Since none of the guards in $M$ fit
this description, we have $|D|\leq k$.  For any vertex $v$, the
end-pixel $\psi(v)$ must have been guarded by some point $g\in S$.
This implies that $p(v)$ $s$-guards $g$, which is possible only if
$g$ lies in the bar of $v$, the bar of some neighbour $u$ of $v$, or
the channel (or adjacent spirals) of some edge $(v,u)$.  Hence $g$ gives
rise to a vertex in $D$ that is either $v$ or a neighbour of $v$.
Thus $D$ is a dominating set as required.

\end{appendix}
\fi

\end{document}